%% file: paper.tex
\documentclass[a4paper, 10pt, onecolumn]{article}
\usepackage[total={6.2in,9.75in},top=1.2in, left=1.0in, includefoot]{geometry}

 \usepackage[utf8]{inputenc}
 \usepackage{xspace}
\usepackage{amsmath,amsthm,thmtools,thm-restate,booktabs,color,doi,graphicx,latexsym,url,xcolor,xspace,soul}
 \usepackage{amsthm,amsmath,array,colortbl,graphicx,multirow}
 \usepackage{balance}
 \usepackage{amsmath}
 \usepackage[font={footnotesize}]{subcaption}
 \usepackage[font={footnotesize}]{caption}
 \usepackage{breakcites}
 \usepackage{booktabs}
 \usepackage{xcolor}
 \usepackage{colortbl}
 \usepackage{enumitem}
 \usepackage{standalone}
 \usepackage{todonotes}
 \usepackage{wrapfig}
 \usepackage[ruled, lined, linesnumbered, commentsnumbered, noend, ]{algorithm2e}
 \usepackage[noend]{algpseudocode}
\usepackage[]{authblk}

 \mathchardef\mhyphen="2D

\input{macros.tex}

\begin{document}

\title{Self-Adjusting Packet Classification}

\author[1]{Maciej Pacut}
\author[1]{Juan Vanerio}
\author[1]{Vamsi Addanki}
\author[1]{\\Arash Pourdamghani}
\author[2]{Gabor Retvari}
\author[1]{Stefan Schmid}

\renewcommand\Affilfont{\fontsize{9}{10.8}\itshape}

\affil[1]{ Faculty of Computer Science, University of Vienna, Austria}
\affil[2]{ Budapest University of Technology and Economics, Hungary}

\date{}

\maketitle

\begin{abstract}
This paper is motivated by the vision of more efficient packet classification mechanisms that self-optimize in a~demand-aware manner.
At the heart of our approach lies a~self-adjusting linear list data structure, where unlike in the classic data structure,
there are dependencies, and some items must be in front of the others; for example, to correctly classify packets by rules arranged in a~linked list,
 each rule must be in front of lower priority rules that overlap with it.
After each access we
can rearrange the list, similarly to Move-To-Front, but dependencies need to be respected.

We present a~4-competitive online rearrangement algorithm, whose cost is at most four times worse than the optimal offline algorithm; no deterministic algorithm can be better than 3-competitive.
The algorithm is simple and attractive, especially for memory-limited systems, as it does not require any additional memory (e.g.,~neither timestamps nor frequency statistics).
Our approach can simply be deployed as a~drop-in replacement for a~static datastructure, potentially benefitting many existing networks.

We evaluate our self-adjusting list packet classifier on realistic ruleset and traffic instances. We find that our classifier performs similarly to a~static list for low-locality traffic, but significantly outperforms Efficuts (by a~factor $7$x), CutSplit ($3.6$x), and the static list ($14$x) for high locality and small rulesets.
Memory consumption is~10x lower on average compared to Efficuts and CutSplit.
\end{abstract}

\section{Introduction}

Packet classification~\cite{gupta2001algorithms} is the most fundamental task in communication networks, performed by switches, routers, and middleboxes (e.g., firewalls).
Existing packet classifiers are typically traffic-oblivious: their internal data structures do not depend on the specific traffic patterns they serve.
For example, decision trees are typically designed for good performance under uniform traffic patterns.

This paper is motivated by the hypothesis that the performance of packet classification can be improved by more adaptive approaches. Indeed, in practice, workloads are often far from uniform, and the majority of traffic can be served with few rules~\cite{SarrarUFSH12}.
If a~packet classifier could identify and optimize toward the ``heavy hitter rules'', its reaction time and throughput could potentially be improved.

Realizing the vision of such self-adjusting packet classifiers, however, is challenging. In particular, as shifts in the traffic patterns (and hence also in the heavy hitter rules) may not be predictable, online algorithms need to be designed which strike a~good tradeoff between the benefits and the costs of rearrangements. Ideally, these online algorithms achieve a~good \emph{competitive ratio}: intuitively, without knowing the future demand, their performance is almost as good as a~clairvoyant optimal offline algorithm.
Besides the challenge of continuously identifying what to optimize for, these algorithms should also be efficient in terms of memory usage and runtime, and not require collecting much statistical data (e.g., related to the hit frequencies) or have high convergence times.

Our approach in this paper builds upon the concept of self-adjusting data structures. Intuitively self-adjusting data structures provide the desired adaptability to workloads, and allow exploiting ``locality of reference'', i.e., the tendency to repeatedly access the same set of items over short periods of time.
Such data structures have been studied intensively for several decades already, including self-adjusting linear lists, which were one of the first data structures of this kind~\cite{sleator1985amortized,SleatorT85}.

\pagebreak
Self-adjusting packet classifiers, however, also introduce novel requirements that do not exist in data structures.
To correctly classify a~packet according to a~certain rule, we must not only check if it is a~\emph{match}, but also if it is the \emph{best match}, by excluding matches to higher priority rules first.
This check is unnecessary if the higher priority rule is independent: if the matching domains of the rules do not overlap, examining the rules in any order determines the action uniquely.
We refer to this challenge as \emph{inter-rule dependencies}. Note that usually, we have a few of them~\cite{KoganNRCE14}.
Consequently, any self-adjusting packet classifier must examine the matching conditions in a~specified order.
This introduces constraints in the internal structure of the classifier, specific to its design.

At the core of our approach to design self-adjusting packet classifiers lies a~self-adjusting linear list data structure.
We store rules in the list's nodes, and to classify a~packet, we match it to the rules in that order.
The rule that matched the packet is then moved closer to the head of the list to speed up future matches.
However, the inter-rule dependencies fix the order of some rules in the list:~each rule must be in front of lower priority rules that overlap with it.
This limits the effectiveness of self-adjustments and significantly changes the structure of the problem.

\subsection{Contributions}

We present a~\emph{self-adjusting packet classifier} that automatically optimizes for the traffic it serves.
Our approach is simple and efficient and does not require maintaining timestamps or frequency statistics, as it only relies on simple pointer swaps and comparisons.
It operates without any a~priori assumptions about traffic distribution and with near-instant convergence to (even local) shifts in traffic. The classification time improves with \emph{locality of reference in traffic}: a~tendency to repeatedly match the packet to a~small set of rules over short periods.
Its behavior is inspired by the Move-To-Front algorithm~\cite{sleator1985amortized} that moves the item to the front each time it is accessed.

\vspace{-0.1cm}
\paragraph{Competitive analysis.}

Our approach provides provable guarantees over time, and we present a~competitive analysis~\cite{Borodin1998} accordingly.
Accounting for packet classification needs, we generalize the classic \emph{online list access problem}~\cite{sleator1985amortized} to the variant with precedence constraints among the nodes of the list.
The designed algorithm, Move-Recursively-Forward is 4-competitive with dependencies, reaching the competitiveness of Move-To-Front\footnote{Move-To-Front is 2-competitive in the \emph{free exchange model}, and 4-competitive in the \emph{paid exchange model}~\cite{Reingold1994}. The latter model fits our study better, see the model section (Section~\ref{sec:model}).} known from the classic setting.
No deterministic online algorithm can be better than 3-competitive: a~result carried from the classic list access problem~\cite{Reingold1994}.

\vspace{-0.1cm}
\paragraph{Handling any inter-rule dependencies.}
The inter-rule dependencies may form any partial order, or equivalently any directed acyclic graph (see Figure~\ref{fig:large-dag}).
While most prior work focuses on the case where the dependencies form a~tree, our algorithm can handle any arbitrary dependency DAG.
Handling arbitrary forms of dependencies saves memory and classification time by eliminating the need for rule preprocessing. If we restrict the dependencies to, e.g., trees, we may face an~exponential growth of the ruleset size due to rule splitting during preprocessing~\cite{gupta2001algorithms}.

\vspace{-0.1cm}
\paragraph{Evaluation.}
We compare the system against packet classifiers that build decisions trees based on hierarchical cuts and study whether a self-adjusting list packet classifier can compete with existing packet classifiers in terms of performance.

\paragraph{Applications.}
Our self-adjusting list packet classifier is a~drop-in replacement for static lists and comes with no memory penalty, no significant increase in complexity, and performance improvements under high traffic locality.
It may further enhance existing static lists maintained in the leaves of some hierarchical cut-based packet classifiers.

\subsection{Related Work}

\paragraph*{Online problems with precedence constraints.}
Various online problems were studied in settings with dependencies and precedence constraints.
In scheduling with precedence constraints~\cite{Azar2002}, a~job can only be scheduled after all its predecessors are completed.
In caching with dependencies~\cite{Bienkowski2017} (another online problem motivated by packet classification), an element can be brought into the cache only if all its dependencies are in the cache.

\paragraph{Cost models for online list access.}
There exist interesting results on list access problems with lower list rearrangement costs.
The most popular variant is the \emph{free exchange model}, in which moving an accessed node forward is free.
In such a~setting, the algorithm Move-To-Front is $2$-competitive, and this result is tight~\cite{sleator1985amortized}.
Other papers considered settings in which rearrangements of large portions of the list have linear costs~\cite{Munro2000,Kamali2013}.
The list access problem was also studied under a~\emph{generalized access cost} model~\cite{sleator1985amortized}, where the cost of accessing an $i$-th node is a~general function $f(i)$.
Some papers considered variants of list access with unequal costs of access and rearrangement, but their focus was on the increased cost of rearrangements --- in contrast to our model, where the packet matching (access) is more costly to rearrangements.
The model $P^d$~\cite{Reingold1994} is a~variant of the paid exchange model where the cost of each transposition is $d \ge 1$.
In this model, the best-known algorithm is COUNTER~\cite{Reingold1994}.
Curiously, our practical motivation to packet classification is the first to justify the cost of reconfiguration that is \emph{smaller} than the cost of access and can be viewed as a~generalization of the $P^d$ model for $d < 1$.

\paragraph*{Competitive algorithms for online list access.}
The online list access problem has been studied for decades~\cite{Mccabe1965}, and remains an active field of research~\cite{Albers2020}.
The Move-To-Front algorithm~\cite{sleator1985amortized} is 4-competitive in the \emph{paid exchange model}~\cite{Reingold1994} (the model used in our paper).
Algorithms such as FREQUENCY-COUNT (sorting items according to use frequencies) or TRANSPOSE (on access exchange with the preceding node) are \emph{not} constant-competitive: their competitive ratio grows with the length of the list~\cite{sleator1985amortized}.
No deterministic online algorithm can be better than 3-competitive~\cite{Reingold1994}.
Interestingly, obtaining an optimal 3-competitive algorithm is challenging: even in the setting without precedence constraints, no deterministic algorithm was proven to be better than 4-competitive yet (but a~survey~\cite{KamaliSurvey2013} suggests that an algorithm Move-To-Front-every-other-access may be 3-competitive).
Randomization helps: RANDOM RESET~\cite{Reingold1994} is $\sqrt{7} \approx 2.64$-competitive against an oblivious adversary.
An optimal offline solution is NP-hard to compute~\cite{Ambuhl00}.

\paragraph*{Empirical evaluation of algorithms for online list access.}
Algorithms for list access perform well in practice~\cite{BachrachER02}.
These algorithms were empirically evaluated under inputs exhibiting various forms of locality:
such as generating them according to Zipf's law or Markov distributions~\cite{BachrachER02}.
A notable work in this area introduced the notion of \emph{runs} that provably improve the competitive ratio~\cite{Albers16}.

\paragraph*{Packet classification.}
Various data structures for packet classification were proposed in the literature: lists, tries, hash tables, bit vectors, or decision trees~\cite{gupta2001algorithms,Srinivasan1999,Eppstein2001}, as well as hardware solutions (TCAM).
Packet classifiers are often accompanied by caching systems that provide some adjustability to traffic.
Due to its simplicity, a~linear lookup structure is commonly applied in practice, e.g., in the default firewall suite of the Linux operating system kernel called \texttt{iptables}~\cite{MianoBRBLP19}, the OpenFlow reference switch~\cite{openflow}, and in many QoS classifiers.

\paragraph*{Packet classification with linked lists.}
Online list access with precedent constraints has applications to packet classification.
We are not the first to study self-adjusting list-based packet classifiers.
Existing heuristic solutions typically rely on frequency statistics, and hence may adapt slowly to shifts in traffic.
\begin{itemize}
  \item The Dynamic Rule Reordering algorithm~\cite{Hamed2006} performs a~batch reorganization to a~configuration that is (heuristically) close to optimum in hindsight, using frequency and recency information to assign weights to the rules.
  This greedy solution builds the list by iteratively inserting the next heaviest rule with all its dependencies (weight of which is not considered).
  \item a~local search algorithm~\cite{Zhao04} updates the list on the fly, possibly after each request.
  Neighboring nodes exchange positions if (1) are independent and (2) their weights are inverted. The procedure stops when further neighbor exchanges are impossible.
  Local search algorithms for some optimization problems may get stuck at a~local minimum, and it is challenging to find a~decent starting point for the search.
\end{itemize}
We note that packet classification with self-adjusting lists is NP-hard even without inter-rule dependencies, by a~reduction from \emph{offline list access}~\cite{Ambuhl00}.

\subsection{Organization}

The rest of this article is organized as follows.
Section~\ref{sec:list} introduces the optimization model for \emph{online list access with precedence constraints} and presents the self-adjusting algorithm Move-Recursively-Forward.
Sections~\ref{sec:analysis}~and~\ref{sec:non-uniform} analyze the competitiveness of Move-Recursively-Forward; we dedicate a~section for the basic cost model and one for the generalized model.
Section~\ref{sec:efficient} discusses different implementation strategies leading to different time-memory tradeoffs. In Section~\ref{sec:insertions},
we discuss insertions and deletions: building blocks for rule update mechanisms in the context of competitive analysis.
We evaluate both implementations against hierarchical cut tree-based algorithms (under various ruleset sizes, inter-rule dependencies, and traffic locality) in Section~\ref{sec:empirical} and
conclude and discuss future research directions in
Section~\ref{sec:conclusions}.

\section{Self-adjusting List with Dependencies}
\label{sec:list}

We realize the vision of self-adjusting packet classifiers by organizing match-action rules in a~linked list. a~packet is first matched against the rule at the head of the list, and then, depending on the action passed on to be classified further by the rules down in the list, according to the list order (respecting the precedence constraints).
With this perspective, matching the packet to its rule in a~list can be seen as the \emph{find} operation for a~node in a~list (in this context often called \emph{access}).
The self-adjustments rearrange the nodes after each access.
We speed up repeated matches to the same rule by moving it closer to the head of the list.
This rearrangement comes at a~cost of swapping pointers.
As described so far, the model is equivalent to \emph{online list access}~\cite{sleator1985amortized}.
However, there are special considerations in packet classification that introduce additional challenges.

To account for inter-rule dependencies, keep certain rules in a~fixed relative order, limiting the effectiveness of rearrangements and making the algorithmic problem significantly more challenging.
Besides moving the accessed node forward in the list, we additionally move forward a~carefully chosen set of the dependencies.
This assures the progress can be made even if the accessed node is blocked (a~dependency is directly in front of it).

\begin{figure}[ht]
  \center
  \includegraphics[width=1\textwidth]{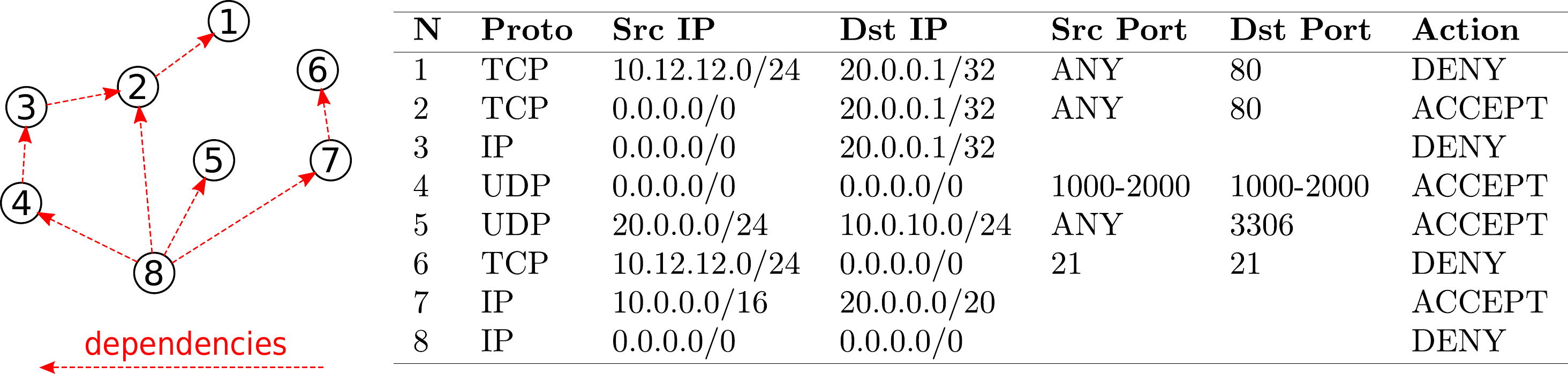}
  \caption{An example of a~table of packet classification rules (right) and the corresponding dependency DAG \DAG among the nodes of the list (left).}
  \label{fig:large-dag}
\end{figure}

\subsection{Optimization Model: Online List Access with Precedence Constraints}
\label{sec:model}

Our task is to manage a~self-adjusting linked list serving a~sequence of requests, with minimal access and rearrangement costs
and accounting for precedence constraints induced by a~directed acyclic graph~\DAG.
If there are no constraints, the problem is equivalent to the classic online list access problem~\cite{sleator1985amortized}.

\paragraph*{The list and the requests}
Consider a~set of $n$ nodes arranged in a~linked list.
Over time, we receive a~sequence $\sigma$ of access requests to nodes of the list.
Upon receiving an access request to a~node in the list, an algorithm searches linearly through the list,
starting from the head of the list, traversing nodes until encountering the accessed node.
Accessing the node at position $i$ in the list costs $i$ (the first node is at position $1$).
Upon receiving a~deletion request, we must access the node first.
Upon receiving an insertion request, we may insert a~node in any feasible position (a list rearrangement may be required to deal with precedence constraints).

\paragraph*{The precedence constraints (dependencies)}
We are given a~directed acyclic graph \DAG of dependency relations between nodes of our list.
This dependency graph induces a~partial order among the nodes (equivalent to the reachability relation in \DAG), and the nodes must obey the partial order in the list in any configuration.
We say that a~node $v$ is a~\emph{dependency} of a~node $u$ if there exists an edge $(u,v)$ in \DAG (we orient edges from the node $u$ to the node $v$ in front of it). Then, in every configuration of the list, $v$ must be in front of $u$.
We assume that the given initial configuration of the nodes obeys the precedence constraints induced by \DAG.

\paragraph*{Node rearrangement}
After serving a~request, an algorithm may choose to rearrange the nodes of the list.
Precisely, the algorithm may perform any number of \emph{feasible} transpositions of neighboring nodes, i.e., transpositions that respect the precedence constraints induced by~\DAG.
We study the \emph{paid exchange model} where all transpositions incur the cost 1; this is different from the
\emph{free exchange model} sometimes considered in the literature that moving the requested node closer to the front is free.

The model poses new algorithmic challenges.
Using algorithms for online list access (such as Move-To-Front~\cite{sleator1985amortized} or Timestamp~\cite{Albers1998}) could lead to infeasible solutions (violating the precedence constraints).
The question arises that whether provably good algorithms even exist for our setting or the problem becomes much harder for online algorithms when faced with dependencies.

\paragraph{Model generalizations} In Section~\ref{sec:non-uniform}, we generalize the model slightly by introducing a~parameter $\alpha\ge 1$ related to the cost of access (a multiplier). In Section~\ref{sec:insertions}, we discuss insertion and deletion operations.

\subsection{Online Algorithms and Competitive Analysis}
\label{sec:competitive}

In this paper the goal is to design competitive algorithms from the model of the previous section.
In this setting, an adversary reveals a~sequence of requests $\sigma$ one-by-one to the online algorithm.
Upon seeing a~request, the algorithm must serve it without the knowledge of future requests. The decisions of the algorithm are permanent: it must be consistent with decisions from the past. For an overview of the competitive analysis framework, we refer to~\cite{Borodin1998}.

\paragraph{Competitive ratio}
We measure
the performance of an online algorithm by comparing its cost with the cost of an optimal offline
algorithm \OPT. Formally, let~$\ALG(\sigma)$, resp.~$\OPT(\sigma)$, be the cost
incurred by a~deterministic online algorithm \ALG, resp.~by an optimal offline
algorithm, for a~given sequence of requests~$\sigma$. In contrast to \ALG, which learns the~requests one-by-one as
it serves them, \OPT has complete knowledge of the entire request
sequence~$\sigma$ \emph{ahead of~time}.
The goal is to design online
algorithms with worst-case guarantees for the ratio. In particular, $\ALG$ is said
to be \emph{$c$-competitive} if there is a~constant~$b$, such that for any
input sequence~$\sigma$ it holds that
\[
	\ALG(\sigma) \leq c \cdot \OPT(\sigma) + b.
\]
Note that $b$ cannot depend on input $\sigma$ but can depend on other
parameters of the problem, such as the number of nodes.
The minimum $c$ for which $\ALG$ is $c$-competitive is called the
\emph{competitive ratio} of $\ALG$.
We say that $\ALG$ is \emph{strictly $c$-competitive} if additionally $b=0$.
All algorithms in this paper are strictly competitive.

The paper mainly focuses on deterministic online algorithms, but in some aspects remarks about \emph{randomized} online algorithms; for an introduction, we refer to~\cite{Borodin1998}.

\subsection{Algorithm Recursively-Move-Forward}
\label{sec:alg}

How to navigate a~DAG of dependencies?
For a~node $y$, we distinguish its \emph{direct dependency}: the dependency node $z$ that would block $y$'s movement forward in the list.
Formally, a~node $y$ is a~\emph{direct dependency} of a~node $z$ if $y$ is the dependency of $z$ that is located at the furthest position in the list.

Our algorithm Move-Recursively-Forward moves forward \emph{the sequence of direct dependencies}: the accessed node is swapped forward until encountering one of its dependencies, at which point the encountered node runs forward recursively.
This strikes a~balance between the access and rearrangement costs: the algorithm exchanges no more pairs than the position of the accessed node (Lemma~\ref{lem:rearrangement-linear}).
In Figure~\ref{fig:follow}, we depict an example run of \MRF after serving a~request.

\begin{figure}[ht]
  \center
  \includegraphics[width=0.8\textwidth]{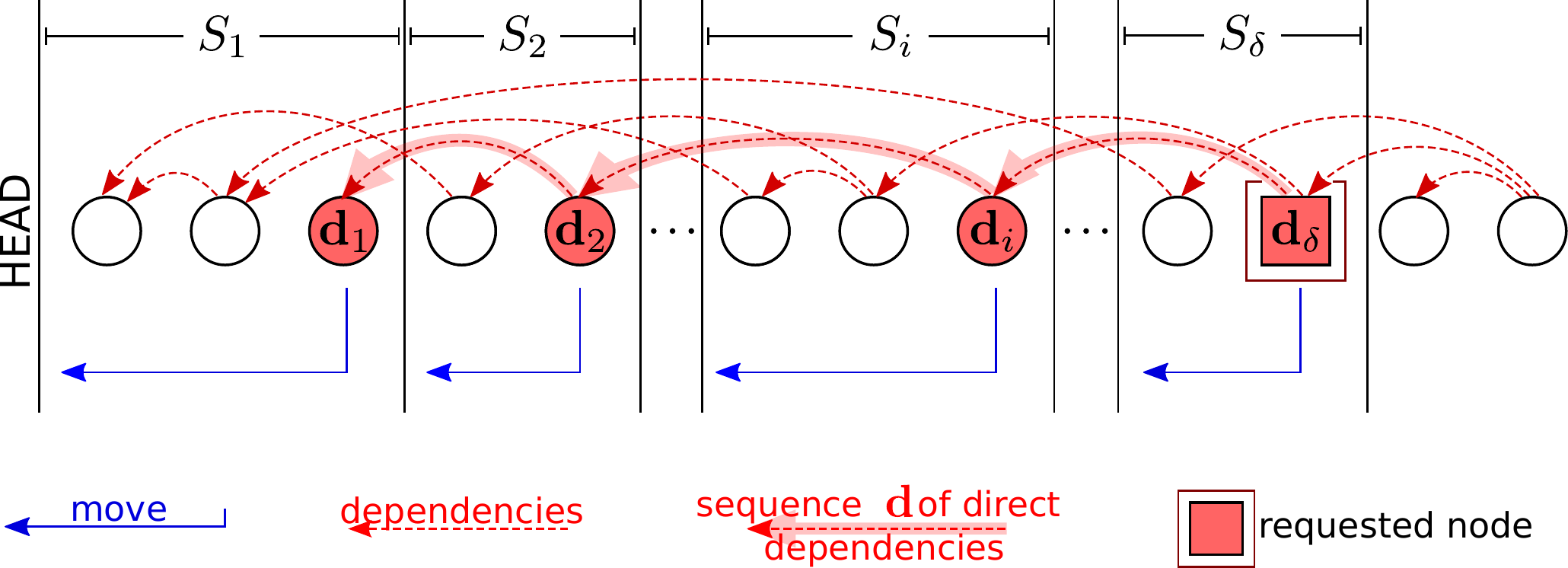}
  \caption{An example of handling a~request by the algorithm \MRF. The red nodes $\dep_i$, formally introduced in Section~\ref{sec:analysis}, are the nodes that the algorithm moves forward (to the position denoted by a~blue arrow). Among multiple dependencies (red dashed arrows), the algorithm \MRF moves forward only the nodes that constitute the path of the direct dependencies (bold red arrow). Sets $S_i$ are used in the analysis, cf. Section~\ref{sec:analysis}.}
  \label{fig:follow}
\end{figure}

\textbf{Algorithm \MRFwd.}
The algorithm uses a~recursive procedure \MRF.
The procedure \MRFx{$y$} moves the node $y$ forward (by transposing it with the preceding nodes) until it encounters any of its dependency nodes, say $z$, and recursively calls \MRFx{$z$}.
Upon receiving an access request to a~node $\sigma_t$, \MRF locates $\sigma_t$ in the list and invokes the procedure \MRFx{$\sigma_t$}.

\medskip
We present the pseudocode of \MRFwd in Algorithm~\ref{alg:follow}.
By $\pos(z)$ we denote the position of node $z$ in the list maintained by the algorithm, counting from the head of the list (recall that the position of the first node is $1$).

\medskip
\begin{algorithm}[h]
 \caption{The algorithm \MRFwd.}
 \label{alg:follow}
\SetAlgoNoLine
\SetKwInOut{Input}{Input}

\Input{An access request to node $\sigma_t$}

 Access $\sigma_t$

 Run the procedure \MRFx{$\sigma_t$}\\

 \BlankLine
  \BlankLine

\SetKwInput{Procedure}{Procedure \MRFx{$y$}}
\Procedure{}
\Indp
 	\uIf{$\text{\upshape $y$ has no dependencies}$}{
		Move $y$ to the front of the list
	}
 	\uElse{
		Let $z$ be the direct dependency of $y$\\
		Move node $y$ to $\pos(z)+1$\\
		Run the procedure \MRFx{$z$}
 }
\end{algorithm}

\medskip

The algorithm can be viewed as a~generalization of Move-To-Front~\cite{sleator1985amortized} to the setting with dependencies. If the dependency graph is empty, the algorithm is equivalent to Move-To-Front.

In the next section, we show that Move-Recursively-Forward is 4-competitive, meaning it performs close to an optimal offline algorithm that dynamically updates its list while knowing the future requests.
The algorithm balances the overhead of adjustments with its benefits:
for each request, the access and the succeeding rearrangements cost roughly the same.

\begin{restatable}{lemma}{rearrangementLinearLemma}
  \label{lem:rearrangement-linear}
  Consider a~single request to a~node $y$ at position $\pos(y)$ handled by the algorithm \MRF.
  The rearrangements after serving the request to $y$ costs at most $\pos(y)$.
\end{restatable}

\begin{proof}
  Intuitively, each of the nodes that were moved forward goes through a~disjoint part of the list, in total at most $\pos(y)$.
  For~a~graphical argument, see Figure~\ref{fig:follow}.

  Recall that $y$ is moved to the position right before its furthest dependency, and then recursively, the dependency moves forward until encountering the dependency of its own.
  The movements end when a~moving node reaches the front of the list.
  Each node is moved right to a~position \emph{one place behind} its dependency.

  Let~$\dep$ be the sequence of the nodes that the algorithm moves forward (calls to \MRF), ordered by increasing distance to the head.
  Let $\depx$ be the length of~$\dep$.
  Then, the node $\dep_i$ moves to the position $\pos(\dep_{i-1})+1$ (for formality of the argument, we assume an artificial dependency at the head of the list, $\pos(\dep_0) = 0$).
  Thus, the total number of transpositions is
  \[
    \sum_i^\depx (\pos(\dep_i) - (\pos(\dep_{i-1})+1)) = \pos(\dep_\depx) - \pos(\dep_0) - \depx \le \pos(\dep_\depx).
  \]
  As $y = \dep_\depx$, we conclude that the lemma holds.
\end{proof}

\section{Competitive Analysis of Move-Recursively-Forward}
\label{sec:analysis}

The competitive ratio compares the online algorithm \ALG with an optimal offline algorithm \OPT: the goal is to minimize $\ALG / \OPT$.
We overview the framework of competitive analysis in Appendix~\ref{sec:competitive}.

In this section, we prove 4-competitiveness of Move-Recursively-Forward for the case where $\alpha=1$ (access cost equal to reconfiguration cost).
The competitive ratio is strict, and we note that the algorithm satisfies the definition of a~\emph{memoryless online algorithm}~\cite{Chrobak1991}.
Later in Section~\ref{sec:non-uniform} we generalize the analysis for arbitrary $\alpha$ and show that the algorithm is asymptotically optimal in this setting too.
In both cases, we use a~potential function argument, and we overview it below.

\paragraph*{Analysis overview.}

The proof schema generalizes the argument of Sleator and Tarjan for Move-To-Front~\cite{sleator1985amortized}.
We fix an optimal offline algorithm \OPT and its run on a~given input $\sigma$, and we relate \OPT's run with the online algorithm's run.
A potential function measures the distance between the algorithms, measured in \emph{inversions}: pairs of nodes in reverse order in Move-Recursively-Forward's and \OPT's lists.
Generalization of the proof required a~careful analysis of inversions caused by movements of many nodes after each access.

\paragraph*{Events overview.}
To analyze the competitiveness on $\sigma$, we sum an amortized cost of a~sequence of events of type:
\begin{enumerate}[label=(\Alph*)]
  \item
  An \emph{access event} $R^i(\sigma_t)$ for $i \in \{0, 1\}$. The algorithm serves the access to the node~$\sigma_t$ and runs the Move-Recursively-Forward procedure.
  We assume a~fixed configuration of \OPT throughout this event.
  \item a~\emph{paid exchange event} of OPT, $P(\sigma_t)$, a~single paid transposition performed by \OPT, where it either creates or destroys a~single inversion with respect to the node~$\sigma_t$.
  We assume a~fixed configuration of \MRF throughout this event.
\end{enumerate}

\subsection{Preliminaries and Notation}
\label{sec:preliminaries}
Before analyzing the competitive ratio of the algorithm, we introduce the notation and the sets and sequences of nodes relevant to our analysis.
\paragraph*{Inversions}
An \emph{inversion} is an ordered pair of nodes $(u, v)$ such that $u$ is located before $v$ in \MRF's list and $u$ is located after $v$ in \OPT's list.
The inversion is the central concept in the analysis of the presented algorithms in this paper.
\paragraph*{The rearranged nodes $\dep_j$}
Consider an access to a~node $\sigma_t$ and the node rearrangements at $t$.
Let~$\dep$ be the sequence of the nodes that the algorithm moves forward (calls to \MRF), ordered by increasing distance to the head.
Let $\depx$ be the length of~$\dep$.
We~emphasize that $\dep$ contains the accessed node at the last position, $\sigma_t = \dep_\depx$.
\paragraph*{Values $k$ and $\ell$}
To compare the costs of \MRF and \OPT, we define values $k$ and $\ell$ related to the number of nodes in front of the accessed node $\sigma_t$ in \MRF's and \OPT's list. Precisely,
 let $k$
be the number of nodes before $\sigma_t$ in both \MRF's and \OPT's lists,
and let $\ell$ be the number of nodes before $\sigma_t$ in \MRF's list, but after $\sigma_t$ in \OPT's list.
\paragraph*{Sets $K_j$ and $L_j$}
With the values $k$ and $\ell$, it is possible to analyze the classic algorithm Move-To-Front, but they are not sufficient to express the complexity of Move-Recursively-Forward.
Hence, we generalize the notion of $k$ and $\ell$ to sets of elements related to positions of individual nodes $\dep_j$ in \MRF's and \OPT's lists.
Precisely, let $K_j$ be the set of elements before $\dep_j$ in both \MRF's and \OPT's lists for $j \in [1, \depx]$,
  and let $L_j$ be the set of elements before $\dep_j$ in \MRF's list but after $\dep_j$ in \OPT's list.
We note that these sets are generalizations of $k$ and $\ell$: for the accessed node $\dep_\depx$ we have $k = |K_\depx|$ and $\ell = |L_\depx|$.
\paragraph*{Sets $S_j$}
The sets of nodes between the nodes $\dep$ in \MRF's list are crucial to the analysis.
Intuitively, the node $\dep_i$ moves in front of all the nodes from the set $S_i$.
Let $S_1$ be the elements between the head of \MRF's list and $\dep_1$ (included).
For $j \in [2, \depx]$, let $S_j$ be the set of elements between $\dep_j$ and $\dep_{j-1}$ (with $\dep_{j-1}$ excluded) in \MRF's list.

\medskip

Figure~\ref{fig:sets} illustrates an example of possible composition of sets $K_j$, $L_j$ and $S_j$ for different values of $j$ on a~given access request.
\begin{figure}[ht]
  \center
  \includegraphics[width=0.6\textwidth]{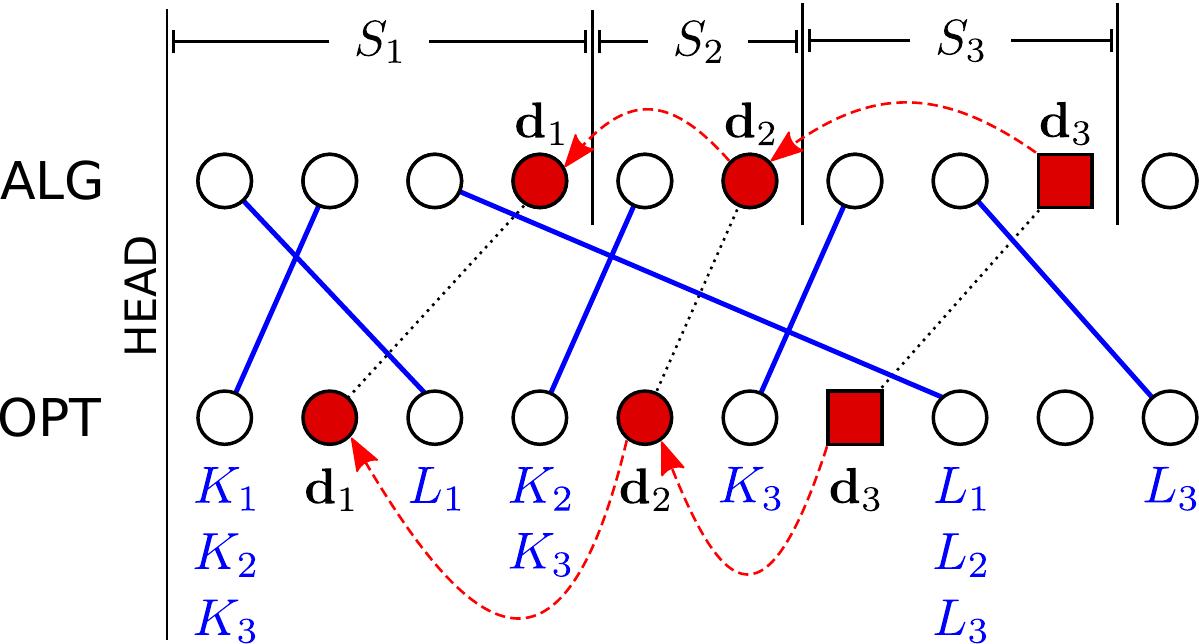}
  \caption{The example illustrates central definitions of sets of nodes used in our analysis. We~depict the positions of nodes in both \MRF's and \OPT's list (joined by solid blue lines). The dotted black lines between the nodes $\dep_i$ help in determining the assignment of nodes to sets: in $K_i$ we have the nodes in front of the dotted black line between $\dep_i$, and in $L_i$ we roughly have the nodes that cross the dotted black lines between $\dep_i$'s.}
  \label{fig:sets}
\end{figure}

\subsection{How Do Rearrangements Change Inversions?}
\label{sec:bounding-inversions}

Consider an access event and the sequence of rearrangements of Move-Recursively-Forward.
We study the influence of the rearrangements on the change in the potential function.
To this end, we separately bound the number of introduced and destroyed inversions.
The Figure~\ref{fig:sets} assists in illustrating the graphical arguments used in this section.

\begin{theorem}
  \label{thm:change-inversions}
  Consider access to the node $\sigma_t$, and fix a~configuration of \OPT at time~$t$.
  Then, the change in the number of inversions due to \MRF's node rearrangement after the access is at~most~$k-\ell$.
\end{theorem}

To prove this claim, we consider the influence of the Move-Recursively-Forward operation on values $k$ and~$\ell$ (defined for the currently accessed node) by inspecting the sets $K_j$ and $L_j$ (defined for the nodes $\dep_j$).
We separately bound the number of inversions created (Lemma~\ref{lem:created-inversions}) and destroyed (Lemma~\ref{lem:destroyed-inversions}).
Before showing these claims, we inspect the basic relations between the sets $K_j$, $L_j$ and $S_j$ (Lemma~\ref{lem:set-relations}).

\pagebreak
\begin{restatable}{lemma}{setRelationsLemma}
  \label{lem:set-relations}
  The following relations hold:
  \begin{enumerate}
    \item \label{it:ki}
    $
    \bigcup_{j=1}^{\depx} K_j = K_{\depx},
    $
    \item \label{it:li}
    $
    \bigcup_{j=1}^{\depx} (S_j \cap L_j) = \bigcup_{j=1}^{\depx} L_j .
    $
  \end{enumerate}
\end{restatable}

\begin{proof}
  First, we prove the equality \ref{it:ki}.
We show inclusions both ways.
Note that the order between nodes from $\dep$ is the same in both \MRF's and \OPT's lists.
Hence, a~node $y \in K_i$ is in front of $\dep_i$ and in front of all $\dep_j$ for $i \le j$ in both \MRF's and \OPT's list.
Consequently, each~node from $K_i$ belongs to all $K_j$ for $i \le j$, and we have $\bigcup_{j=1}^{\depx} K_j \subseteq K_{\depx}$.

Conversely, $K_{\depx} \subseteq \bigcup_{j=1}^{\depx} K_j$ by basic properties of sets, and we conclude that the sets are equal, and the equality holds.

\medskip

  Next, we prove the equality \ref{it:li}.
  We show inclusion both ways.
  Consider any element $y \in L_j$.
  The sets $\{S_j\}$ partition the nodes placed closer to the front of the list than $\sigma_t$ (the requested node), thus $y$ belongs to some $S_i$ for $i \le j$.
  Fix such $i$; we claim that additionally $y \in L_i$:
  \begin{itemize}
    \item $y$ belongs to $S_i$, and hence it is in front of $\dep_i$ in \MRF's list,
    \item $y$ is after $\dep_j$ in \OPT's list (it belongs to $L_j$), and hence it is after $\dep_i$ in \OPT's list (the order of $\dep$ is fixed due to dependencies).
  \end{itemize}
  Hence, any $y \in L_j$ belongs to $S_i \cap L_i$ for some $i$, and we conclude that the inclusion $\bigcup_{j=1}^{\depx} L_j \subseteq \bigcup_{j=1}^{\depx} (S_j \cap L_j)$ holds.
  Conversely, by properties of sets $\bigcup_{j=1}^{\depx} (S_j \cap L_j) \subseteq \bigcup_{j=1}^{\depx} L_j$, and we conclude that the sets are equal and the equality holds.
\end{proof}

\begin{lemma}
  \label{lem:created-inversions}
  Consider an access to a~node $\sigma_t$, and fix a~configuration of \OPT at time $t$.
  Due to rearrangements after the access, \MRF creates at most $k$ inversions.
\end{lemma}

\begin{proof}
  Let $\Ic_j$ be the number of inversions added by moving a~single node $\dep_j$ by \MRF, for $j \in [1, \depx]$.
  To bound $\Ic_j$, we inspect the set $S_j$ of the nodes that $\dep_j$ overtakes, and we reason based on their positions in \OPT's list.
  Moving $\dep_j$ forward creates inversions with nodes in (possibly a~subset of) $S_j \cap K_j$.
  No other node changes its relation to the set $S_j$, hence the inversions for nodes in $S_j$ are influenced only by the movement of $\dep_j$.
  This gives us the bound
  $ \Ic_j \leq | S_j \cap K_j |$.%

  We sum up the individual bounds on $\Ic_j$ for all $j$ to bound the total number of inversions created%

$$ \sum_{j=1}^{\depx} \Ic_j \leq \sum_{j=1}^{\depx} | S_j \cap K_j | = \lvert \bigcup_{j=1}^{\depx} (S_j \cap K_j) \rvert \le | \bigcup_{j=1}^{\depx} K_j |,$$
where the second step holds as the sets $S_j$ are disjoint, and the last step follows by basic properties of sets.%
By Lemma~\ref{lem:set-relations}, equation \ref{it:ki}, we have
$|\bigcup_{j=1}^{\depx} K_j | = |K_{\depx}|$, thus by combining the above inequalities, we have $ \sum_{j=1}^{\depx} \Ic_j \leq |K_{\depx}| = k$, and we conclude that the claim holds.
\end{proof}

\begin{lemma}
  \label{lem:destroyed-inversions}
  Consider an access to the node $\sigma_t$, and fix a~configuration of \OPT at time $t$.
  Due to rearrangements after the access, \MRF destroys at least $\ell$ inversions.
\end{lemma}

\begin{proof}
  Let $\Id_j$ be the number of inversions destroyed by moving a~single node $\dep_j$ by \MRF, for $j \in [1, \depx]$.
  To bound $\Id_j$, we inspect the set $S_j$ of the nodes that $\dep_j$ overtakes, and we reason based on their positions in \OPT's list.
  Moving $\dep_j$ forward destroys all inversions with nodes in $S_j \cap L_j$.
  No other node changes its relation to the set $S_j$, hence the inversions for nodes in $S_j$ are influenced only by the movement of $\dep_j$.
  This gives us the bound
  $ \Id_j \ge | S_j \cap L_j |$.%

  We sum up the individual bounds on $\Id_j$ for all $j$ to bound the total number of inversions destroyed.%
  $$ \sum_{j=1}^{\depx} \Id_j \geq  \sum_{j=1}^{\depx} | S_j \cap L_j | = \lvert \bigcup_{j=1}^{\depx} (S_j \cap L_j) \rvert,$$
  where the second step holds as the sets $S_j$ are disjoint.%
  By Lemma~\ref{lem:set-relations}, equation~\ref{it:li} we have
  $|\bigcup_{j=1}^{\depx} (S_j \cap L_j)| = |\bigcup_{j=1}^{\depx} L_j |$.
  Finally, by basic properties of sets, $ | \bigcup_{j=1}^{\depx} L_j | \ge |L_{\depx}| = \ell$,
  and by combining all the above bounds we have $\sum_{j=1}^{\depx} \Id_j \geq  \ell$,
  and we conclude that the lemma holds.
\end{proof}

Combining Lemmas~\ref{lem:created-inversions} and \ref{lem:destroyed-inversions} gives us the joint bound on the change in the number of inversions and proves the Theorem~\ref{thm:change-inversions}.
We note that this bound is consistent with the bound on the changes in inversions for the algorithm Move-To-Front~\cite{sleator1985amortized}, where the inversions were considered with respect to the accessed node only.

\subsection{Bounding the Competitive Ratio}
\label{ssec:ratio}

Finally, we show the main result of this section: the competitive ratio of the algorithm \MRF is 4 for $\alpha=1$.
First, we apply the bounds on the number of inversions (Theorem~\ref{thm:change-inversions}), and then we bound the ratio using a~potential function argument (Theorem~\ref{thm:follow-4-competitive}).%

\paragraph*{Potential function.}
We define the potential function $\Phi$ in terms of the number of inversions in \MRF's list with respect to OPT's list.
Precisely, the potential function is defined as \emph{twice the number of inversions}.
The factor $2$ accounts for the cost of the access and the exchange for each inversion.

Finally, we prove the main result of this section, that the algorithm \MRF is $4$-competitive.
The proof uses a~potential function argument similar to Move-To-Front~\cite{sleator1985amortized}, and it internally uses the Theorem~\ref{thm:change-inversions} to reason about changes in inversion due to Move-Recursively-Forward runs.

\begin{restatable}{theorem}{followCompetitive}
  \label{thm:follow-4-competitive}
  The algorithm \MRF is strictly 4-competitive.
\end{restatable}

\begin{proof}
  Fix a~sequence of access requests $\sigma$.
  We compare the costs of \MRF and an optimal offline algorithm \OPT on $\sigma$ using a~potential function $\Phi$.
  Let $C_{\MRF}(t)$ and $C_{\OPT}(t)$ denote the cost incurred at time $t$ by \MRF and \OPT respectively.

  First, we bound the cost of \MRF incurred while serving access to a~node $\sigma_t$ at time $t$ (an \emph{access event}).
  This cost consists of the access cost and the rearrangement cost.
  To access the node $\sigma_t$, the algorithm incurs the cost $\pos(\sigma_t)$, and by Lemma~\ref{lem:rearrangement-linear} the rearrangement cost is bounded by $\pos(\sigma_t)$, hence
  $C_{\MRF}(t) \le 2\cdot \pos(\sigma_t)$.

  Next, we bound the amortized cost for every access request served by \MRF.
  The amortized cost is $C_{\MRF}(t) + \Delta \Phi(t)$ for each time $t$.
  By Theorem~\ref{thm:change-inversions}, we bound the change in the number of inversions due to \MRF's rearrangement after serving the request at time $t$ by $\Delta I \le k - \ell$.
  Thus, the change in the potential is $\Delta \Phi(t) \le 2(k - \ell)$.
  As~$\pos(\sigma_t) = k + \ell + 1$, combining these bounds gives us
  \begin{align*}
    C_{\MRF}(t) + \Delta \Phi(t) \le & 2\cdot \pos(\sigma_t) + 2 \left( k - \ell \right) \\
    \le & 2 \left(k+\ell+1\right)+ 2 \left(k - \ell \right)
    \le 4 \cdot C_{\OPT}(t),
  \end{align*}
  where the last inequality follows by $C_{\OPT} \ge k + 1$.

  Note that the bound on amortized cost accounts for possible \emph{paid exchange events}, the rearrangement of \OPT at time $t$.
  Each transposition of \OPT increases the number of inversions by $1$, which increases the LHS by $2$;
  and for each transposition \OPT pays $1$, which increases the RHS by $4$.

  Finally, we sum up the amortized bounds for all requests of the sequence $\sigma$ of length $m$, obtaining
  \begin{equation*}
  C_{\MRF}(\sigma) + \Phi(m) - \Phi(0)  \le 4\cdot C_{\OPT}(\sigma).
  \end{equation*}
  We assume that \MRF and \OPT started with the same list, thus the initial potential $\Phi(0) = 0$, and the potential is always non-negative, thus in particular $\Phi(m) \ge 0$, and we conclude that
  $C_{\MRF}(\sigma) \le 4\cdot C_{\OPT}(\sigma)$.\qedhere
\end{proof}

\section{Competitive Analysis in the General Cost Model}
\label{sec:non-uniform}

Some algorithms may incur a different cost at access than at rearrangement, for example, due to additional computations needed before swapping pointers (as we will elaborate for the second algorithm from Section~\ref{sec:efficient}).
Such algorithms operate under a~different cost model, and our analysis must be strengthened to account for various ratios of the incurred cost.

\medskip
\pagebreak
\noindent
\textbf{Generalized model.} We introduce a~parameter $\alpha \ge 1$ to the model defined in Section~\ref{sec:model}.
Then,
accessing the node at position $i$ in the list costs $\alpha \cdot i$ (instead of $i$). The first node is at position $i=1$.
\medskip

We repeat the competitive analysis of Move-Recursively-Forward in the generalized model to avoid coming to possibly wrong conclusions of a~different cost model. The results from Section~\ref{sec:analysis} hold for $\alpha = 1$.
The central question of this section is:
\begin{center}
  \emph{Should we redesign our algorithm, and if so, does it need to be aware of $\alpha$?}
\end{center}

This generalization is the opposite of the $P^d$ model~\cite{Reingold1994}, where the cost of rearrangement is higher than access: $d \ge 1$.
Curiously, to the best of our knowledge, the cost model with $d < 1$ was not considered before.
We use the parameter $\alpha = 1/d$ for simplicity.
We note that for the converse scenario --- the standard $P^d$ model with $d > 1$, the algorithms must be modified in a~way that depends on the cost of reconfiguration (e.g., COUNTER algorithm~\cite{Reingold1994}) to remain competitive.

In Theorem~\ref{thm:non-uniform-upper}, we answer the question of this section from the standpoint of competitive analysis:
\begin{center}
  \emph{ \MRF is asymptotically optimal for any value of $\alpha$, as long as $\alpha < n$,}
\end{center}
where $n$ is the number of nodes in the list.
The value of $\alpha$ is the cost of multi-field packet matching and is expected to be a~small constant, far smaller than a~typical list length.

\subsection{Lower Bound}

The competitive analysis framework finds computationally impossible results for $\alpha > 1$: lower bounds linear in $\alpha$ for any online algorithm!
The Theorem~\ref{thm:non-uniform-upper} reveals a~more pessimistic scenario than a~lower bound of $3$ for $\alpha = 1$~\cite{Reingold1994}.
We prove the result for the classic \emph{online list access} without dependencies since the result carries to our model.

\begin{theorem}
  \label{thm:non-uniform-lower-det}
  If an algorithm \ALG is a~deterministic $c$-competitive algorithm for online list access in the general cost model with the visiting cost $\alpha$ and $n$ nodes, then $c \ge \frac{n\cdot \alpha}{n+\alpha}$.
\end{theorem}

\begin{proof}
  Fix any deterministic algorithm \ALG, and consider an input sequence $\sigma$, where $\sigma(t)$ is the request to the last element of \ALG's list at time $t$.
  The cost of traversing a~node is $\alpha$, thus, the cost incurred by \ALG for serving each request is $n\cdot\alpha$.

  Consider the following offline strategy \OPT: before serving a~request to a~node $\sigma(t)$, move it to the front of the list (at time $t-1$).
  The cost of serving each request is then upper-bounded by $n+\alpha$.
  The competitive ratio is then at least $\frac{\ALG}{\OPT} \ge \frac{n\cdot \alpha}{n+\alpha}$.
\end{proof}

We remark that \emph{randomization} does not help more than a~constant factor: the lower bound against randomized online algorithms also follows the linear trend in $\alpha$.
An oblivious adversary~\cite{Borodin1998} cannot access the last node in the randomized algorithm's list, but instead, it may access nodes with uniform probability.
The expected cost of access for the algorithm is roughly $\alpha \cdot n/2$.

\subsection{Upper Bound}

Now we turn to the analysis of Move-Recursively-Forward in the generalized model.

\begin{theorem}
  \label{thm:non-uniform-upper}
  The algorithm \MRF is $\max\{4, 1+\alpha\}$-competitive for list update with dependencies in the general cost model.
\end{theorem}

\begin{proof}
The proof is a~straightforward modification of 4-competitiveness proof from the previous section.
To obtain the result, we use a~potential function $\Phi = (1+\alpha)\cdot I$, where $I$ is the number of inversions.
In case of a~\emph{request event}, the ratio is $4$, and in case of a~\emph{paid exchange event} (by OPT), the ratio is $1+\alpha$.
\end{proof}

We note that the lower bound from Theorem~\ref{thm:non-uniform-lower-det} for $3 \le \alpha \le n$ is at most $3$ times smaller than the upper bound from Theorem~\ref{thm:non-uniform-upper}.
Thus, the algorithm Move-Recursively-Forward is asymptotically optimal in this setting.

\section{Efficient Implementation}
\label{sec:efficient}

So far, we abstracted from the implementation details of the algorithm Move-Recursively-Forward.
The algorithm Move-Recursively-Forward can be implemented without any mutable memory besides the linked list:
the dependency relations may be computed on the fly.
The algorithm can be implemented in an array, reaching the lower bounds for (uncompressed) memory consumption.
The recursive procedure is \emph{tail-recursive}, eliminating the need for a~call stack.

In this section, we consider two approaches for representing the dependencies graph: \emph{a simple memoryless algorithm} that computes dependencies on the fly, and \emph{a more efficient algorithm} that precomputes and stores in-memory the dependency relations between nodes.
This enables a~time-memory tradeoff for the self-adjusting packet classifier.

Storing the dependencies in memory speeds up the execution time of direct dependency operation.
Instead of determining rule overlap for each field, we may look it up in memory and compare pointers or checking if a~bit is 1 (depending on implementation).
However, this increases memory consumption: a~time-memory tradeoff known from other systems (e.g., node replication factor in Efficuts~\cite{VamananVV10}).

Storing the dependencies enables even further optimization options: to minimize the dependency graph while preserving the correctness of the algorithm.
Notice that we may remove the dependencies that are indirect: in a~chain of overlapping rules, we may store only the dependencies between neighbors rather than a~full all-pairs graph.
Consider two dependencies of a~given node C: A (indirect) and B (direct).
Consider A is a~dependency of B too, hence their order is fixed.
Thus, the direct dependency of C can never be A.
The indirect dependencies are never relevant in the direct dependency operation, and we may get rid of them, improving the running time and memory consumption.

\subsection{First Algorithm: Memoryless}
\label{sec:mrfmemoryless}

How to represent the dependency graph in memory?
The simplest solution, the \emph{memoryless} algorithm, avoids storing it altogether: to compute the dependencies on the fly by checking rules overlap and priority. In this implementation, matching a~packet to a~rule costs roughly as much as checking the dependency relation, fitting our model with $\alpha=1$.
This is our baseline implementation, and we evaluate this solution in Section~\ref{sec:empirical}.
This implementation is simple: below 150 lines of C++ code, whereas some packet classifiers such as EffiCut have thousands of lines of code~\cite{VamananVV10}.

\subsection{Second Algorithm: Storing the Dependency Graph in Memory}
\label{sec:mrffast}

We propose an alternative implementation, designed to trade memory for execution time by storing the dependency graph in memory.
By accessing the precomputed dependency relations, we may expect faster execution of the rearrangement logic.

The gains of storing the dependency graph vary with the representation. An adjacency matrix guarantees constant-time lookup but at the cost of a~large memory footprint.
An adjacency list guarantees lookup proportional to the number of neighbors and saves a~substantial amount of memory, given that common dependency graphs are sparse.
A sparse matrix representation shows decent asymptotic runtime, employing binary search.
In practice, solutions based on hashing may perform very well.
We do not pursue the representation question further, but we assume its efficient memory footprint: proportional to the number of stored dependencies.

Storing all edges of the dependency graph is not necessary, as it can even be detrimental to the algorithm's performance.
Consider a~chain of $m$ rules, where the lower priority rules overlap with higher priority rules (inducing the dependency relation).
In the dependency graph, the last rule has $m-1$ dependencies, and the total number of edges is in order of $m^2$.
The graph contains not only edges between neighbors but also its \emph{transitive closure}.

Instead of storing all edges, a~\emph{minimum equivalent DAG}~\cite{MoylesT69} is sufficient.
Removing indirect dependencies speeds up the direct dependency operation significantly.
Moreover, reducing the graph affects the memory consumption, assuming memory-efficient representations.

\pagebreak
At this point, the cost model changes, and we generalize the competitive analysis to the updated setting, finding that the algorithm is still near-optimal.
The constant $\alpha$, defined in Section~\ref{sec:model} is larger than $1$, because the rearrangement cost (swapping) becomes cheaper than before: it was roughly the cost of matching the packet, as it required determining dependencies. Now, it consists only of a~pointer comparison and swap, while matching a~packet to a~rule needs to examine multiple fields, yielding $\alpha > 1$.

\subsection{Competitive Analysis Remarks}

A question arises around various rearrangement strategies available: if the algorithm balances the rearrangement cost, should we rearrange the list more aggressively?
For example, by moving more (or all) dependencies forward, we may even further improve the position of the accessed node for future accesses.
This would be ill-advised, as such an algorithm is less competitive: moving all to the front is $O(d)$-competitive, where $d$ is the maximum depth of the dependency graph.

What would be the effect of rearranging less often?
For example, to rearrange the list after every 10th access, or with probability $1/10$.
These more conservative algorithms were analyzed for the classic list access (without dependencies), and were designed for a~different setting, where the cost of an exchange is \emph{larger} than the access cost.
In this model, the best-known algorithm is COUNTER~\cite{Reingold1994}.
However, such modifications are unnecessary:
in Section~\ref{sec:non-uniform}, we claim that the algorithm Move-Recursively-Forward remains near-optimal for small $\alpha$.

\section{Handling Insertions and Deletions}
\label{sec:insertions}

In data structures, such as linked lists, the sets of nodes can change over time.
Hence, many data structures (including the ones for packet classification) also support insertions and deletions in addition to access operations.

To our best knowledge, we are the first to explore the problem of insertions and deletions in list data structures with dependencies.
This section discusses the feasibility of insertions, as an algorithm may require reconfiguring other rules before inserting a~new rule.
We propose assumptions that allow for constant competitiveness with insertions.

\subsection{Model: Online List Update with Precedence Constraints}
Consider the online list \emph{update} problem with dependencies with three request types: \emph{accesses} to existing nodes in the list,
\emph{insertions} of new nodes and \emph{deletions} of existing nodes in the list.
Upon receiving an insertion request, the node reveals its dependencies with the nodes that are already in the list.
The revealed dependencies must be obeyed until a~node is deleted.
Note that a~node may have dependencies with nodes that will be inserted later, but this information is unknown until then.

\subsection{Packet Classification Challenges: Non-Transitivity}

The desired property of possible dependency graphs is \emph{transitivity}: the edges of the graph induce a~relation that is a~\emph{transitive closure}.
We report that the packet classification rule dependencies \emph{may not} have the transitivity property.
An inserted rule may reveal dependencies between the existing rules that become violated.
For an example, see Table~\ref{fig:example4_2}.

\begin{table}[!h]
  \centering
\begin{tabular}{lllllll}
\hline
\textbf{N} & \textbf{Proto} & \textbf{Src IP} & \textbf{Dst IP} & \textbf{Src Port} & \textbf{Dst Port} & \textbf{Action} \\ \hline
1 & TCP & 10.1.1.1 & 20.1.1.1 & ANY & 80 & ACCEPT \\
2 & TCP & 10.1.1.2 & 20.1.1.1 & ANY & 80 & ACCEPT \\
3 & TCP & 10.1.1.3 & 20.1.1.1 & ANY & 80 & ACCEPT \\
x & TCP & 10.1.1.0/24 & 20.1.1.1 & ANY & ANY & DENY \\
4 & TCP & 0.0.0.0/0 & 0.0.0.0/0 & ANY & 445 & ACCEPT \\
5 & TCP & 0.0.0.0/0 & 0.0.0.0/0 & ANY & 17 & ACCEPT \\
6 & TCP & 0.0.0.0/0 & 0.0.0.0/0 & ANY & 18 & ACCEPT \\
\hline
\end{tabular}
\caption{
Consider a~set of $n$ nodes with no dependencies between them.
Then, consider an insertion of a~rule $x$, which reveals the dependencies with all existing rules.
The insertion of rule $x$ enforces the order between all existing rules.
The dependencies imposed by the rules are not transitive. With~the transitivity assumption, the rules (4-6) would have dependencies to rules (1-3), and the rule $x$ would be inserted between them without rearrangements. The situation is illustrated in Figure~\ref{fig:example4_1}.}
  \label{fig:example4_2}

\end{table}

Lack of transitivity may lead to costly rearrangements required to insert a~single node.
Figure~\ref{fig:example4_1} depicts an example.

How to deal with insertions without preprocessing rules?
Faced with the need to rearrange before insertion, we recommend finding the minimum rearrangement that satisfies new constraints. Examples such as Figure~\ref{fig:example4_1} cannot happen too often: each time they add dependencies, and the adversary would need to delete nodes to deceive the algorithm again. An amortized analysis may lead to constant-competitiveness, and we leave it to future work.

\begin{figure}[ht]
  \center
  \includegraphics[width=1\textwidth]{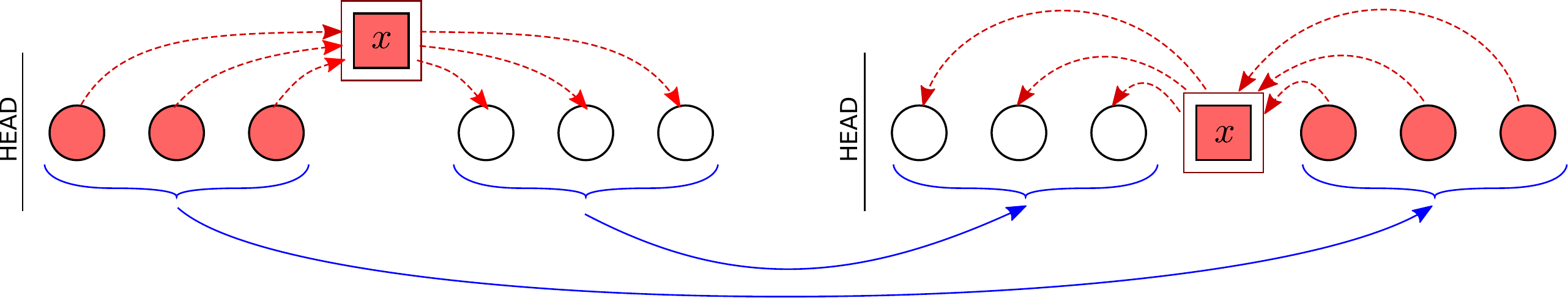}
  \caption{
  Consider a~set of $n$ nodes in a~list, with no dependencies between them.
Then, consider the insertion of a~node $x$, which reveals the depicted dependencies (arrow points to the node that must precede in the list). Here, $x$ has
$\frac{n}{2}$ dependencies from the first half list (the red nodes) to $x$, and
$\frac{n}{2}$ dependencies from $x$ to the nodes of the second half of the list (the white nodes).
To insert $x$ without violating the dependencies, the nodes on the first half of the list must move ahead of the nodes on the second; thus, the cost of rearrangement is $(\frac{n}{2})^2 $.
  }
  \label{ex:general}
  \label{fig:example4_1}
\end{figure}

\subsection{Constant-competitiveness With Transitivity Assumption}

We build on assumptions present in the online list update problem~\cite{sleator1985amortized},
where insertions cost $n$ and nodes are accessed before deletion~\cite{sleator1985amortized}.
Additionally, we require \emph{transititivy}.

Under these assumptions, any algorithm inserts the node $\sigma_t$ at any position of its choice that respects the dependency DAG $G$, \emph{without incurring further cost}.
A position that satisfies all constraints already exists in every configuration due to transitivity assumption.
Consider the universe of all nodes, including the ones that may be inserted in the future.
Assume the transitivity of their precedence constraints, meaning that if a~node $v_1$ must be in front of $v_2$, and $v_2$ must be in front of $v_3$, then $v_1$ must be in front of $v_3$, \emph{even if $v_2$ is currently not present in the list}.

\medskip

\noindent
\textbf{Move-Recursively-Forward with insertions and deletions.}
The algorithm's logic for access does not change.
Upon receiving an insertion request, insert the new node into an arbitrary feasible position (respecting the dependencies).
Upon receiving a~deletion request, the algorithm deletes the node.

\begin{theorem}
    The algorithm Move-Recursively-Forward with insertions and deletions is strictly\\\mbox{4-competitive} for $\alpha=1$, with transitive \DAG, with insertions costing $n$ and accesses before deletes.
  \label{thm:insertions-constant}
\end{theorem}

\begin{proof}
  Fix a~sequence of requests $\sigma$ of length $m$.
  We compare costs of \MRF and an optimal offline algorithm \OPT on $\sigma$ using the potential function $\Phi$, defined \emph{twice the number of inversions} (as before).
  Let $C_{\MRF}(t)$ and $C_{\OPT}(t)$ denote the cost incurred at time $t$ by \MRF and \OPT respectively.

  The transitivity assumption enables insertion without node rearrangements:
  a~feasible position for insertion always exists.
  Thus, for each insertion, both \MRF and \OPT pay exactly $C_{\OPT}(t)=C_{\MRF}(t)=n$.
  Note that \MRF and \OPT might place the node at different positions; thus, we account for at most $n$ inversions created, and consequently, for the change in the potential, we have $\Delta \Phi(t) \leq 2 n $.
 This gives us
    $C_{\MRF}(t)+\Delta \Phi(t) \leq 3 n \leq 3\cdot C_{\OPT}(t)$.

  In case of a~deletion request, exactly $\ell$ inversions are removed, and since $C_{\OPT}(t)=k$ and $C_{\MRF}(t)=k+\ell$, we have
    $C_{\MRF}(t)+\Delta \Phi(t) \leq (k+\ell)-2 \ell \leq k \leq 1 \cdot C_{\OPT}(t)$.

  Similarly to the proof of Theorem~\ref{thm:follow-4-competitive}, the amortized cost for an access request is
  $C_{\MRF}(t) + \Delta \Phi(t) \le 4 \cdot C_{\OPT}(t)$.
  We sum up the amortized bounds for all requests of the sequence $\sigma$, of all types (access, insertion, deletion), and we conclude that \MRF is 4-competitive.
\end{proof}

\section{Empirical evaluation}

\label{sec:empirical}
\label{sec:hypothesis}

We now empirically evaluate the benefits and limitations of our self-adjusting packet classifier by comparing it to existing packet classifiers. We evaluate two implementation variants: MRF (Section~\ref{sec:mrfmemoryless}) and MRF-Fast (Section~\ref{sec:mrffast}). We investigate two questions:
\medskip

\noindent \textit{\textbf{Do self-adjustments improve the classification time?}}
We show that compared to a~static list, the self-adjusting list (MRF) improves the classification time by at least $2$x on average and at least $10$x better under high locality in traffic. Compared to hierarchical cut classifiers, under high locality and for small ruleset sizes, MRF improves classification time.
Specifically, our results show that MRF's classification time is $7.01$x better than Efficuts~\cite{VamananVV10} and $3.64$x better than CutSplit~\cite{li2018cutsplit} for small ruleset sizes and high traffic locality.

\noindent \textit{\textbf{What is the benefit in memory consumption?}}
Our evaluation shows that the simple data structure of the self-adjusting list classifier dramatically improves memory consumption: on average $10$x lower than hierarchical cut classifiers. The implementation MRF-Fast (the second algorithm from Section~\ref{sec:efficient}) uses slightly more memory --- but still lower than any decision-tree-based classifiers. This enables the time-memory tradeoff between MRF and MRF-Fast.

Figures~\ref{fig:heatmap-cutsplit} and \ref{fig:heatmap-efficuts} show the key findings of our evaluation, presenting the average nodes traversed (representing classification time) normalized to CutSplit and Efficuts. We observe that MRF's classification time improves faster than CutSplit and Efficuts in the region of smaller rulesets and high traffic locality. However, MRF's classification time degrades below the competition outside these regions.

\begin{figure*}[h]
\center
\begin{subfigure}{0.48\linewidth}
\includegraphics[width=1\linewidth]{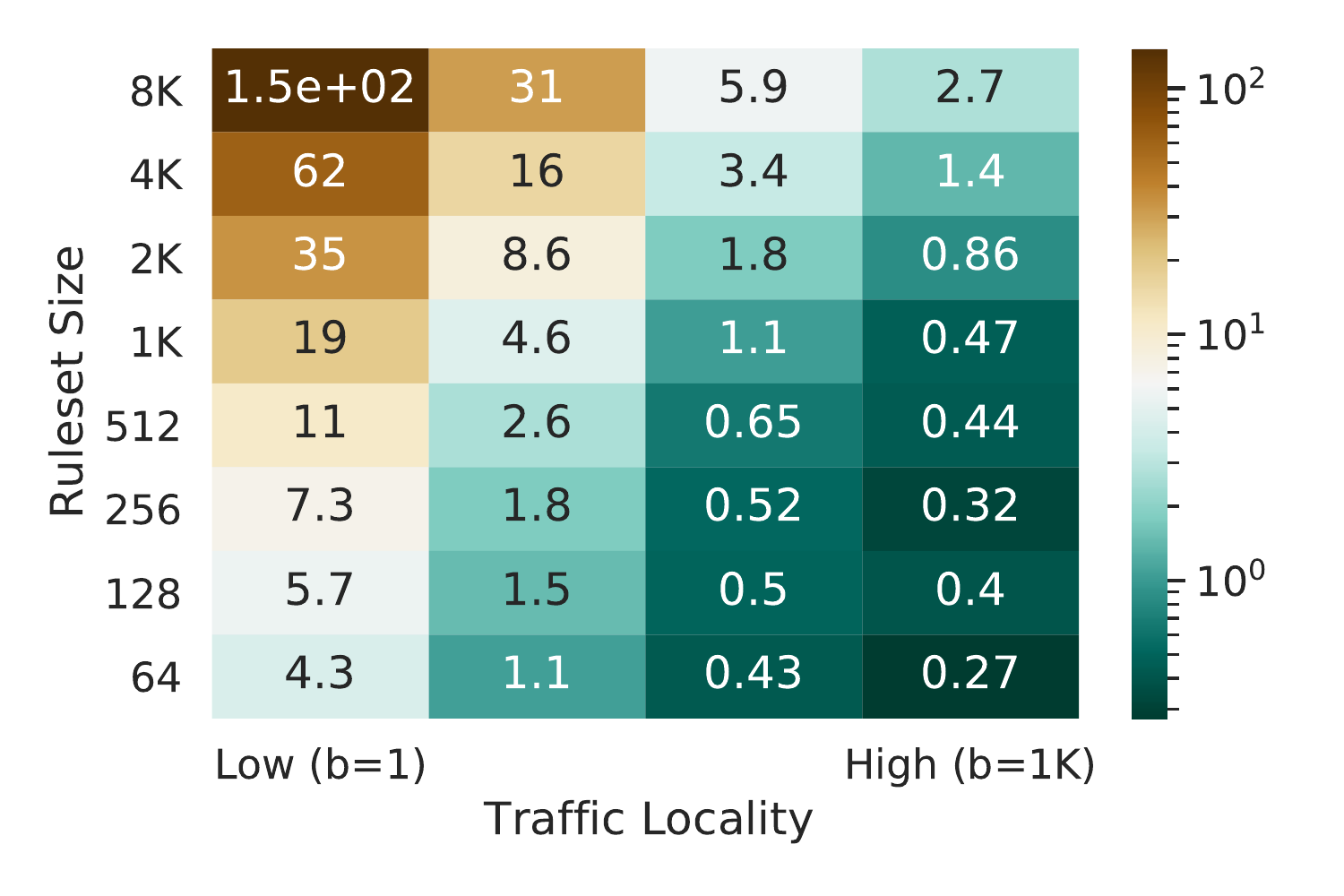}
\subcaption{$\frac{MRF}{CutSplit}$ normalized Avg. nodes traversed.}
\label{fig:heatmap-cutsplit}
\end{subfigure}\hfill
\begin{subfigure}{0.48\linewidth}
\includegraphics[width=1\linewidth]{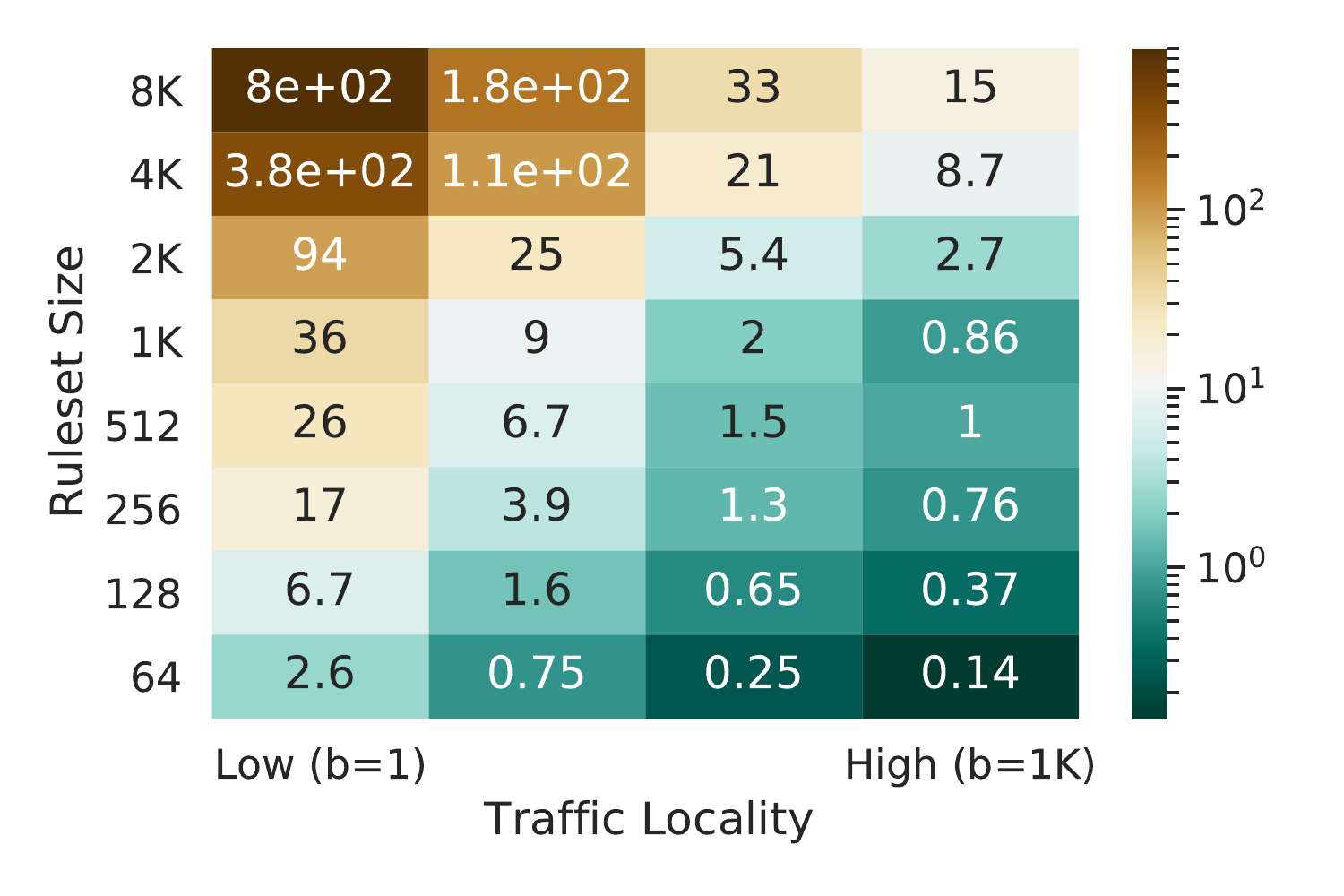}
\subcaption{$\frac{MRF}{Efficuts}$ normalized Avg. nodes traversed.}
\label{fig:heatmap-efficuts}
\end{subfigure}
\caption{MRF out-performs static decision tree-based packet classifiers for small ruleset sizes and at high traffic locality (towards the bottom right region of the ruleset size vs. locality dimensions).}
\end{figure*}

\subsection{Methodology}

\paragraph*{Rulesets and Packet traces:}
We use Classbench~\cite{4237157} to generate rulesets and traffic resembling real-world scenarios.
We examine a~wide range of rulesets sizes and traffic locality. To control traffic locality, we use Classbench's parameter $Pb$: the Pareto distribution $scale$ parameter~\cite{doi:10.1080/00401706.1987.10488243}.

\paragraph*{Comparison with existing packet classifiers:}
We compare MRF and MRF-Fast against list-based and decision tree-based packet classifiers. In our evaluations, a~static-list serves as a~baseline for a list-based approach. Among wide range of decision tree-based packet classifiers, our baselines include Hicuts~\cite{820051}, Hypercuts~\cite{singh2003packet}, Efficuts~\cite{VamananVV10} and the more recent packet classifier CutSplit~\cite{li2018cutsplit} which combines cutting and splitting techniques. We use the default parameter settings for all our baselines.

\paragraph*{Simulations:}
We built a~custom simulator written in C++ and implemented all the baselines, including MRF and MRF-Fast. We have faithfully merged the online available source code of our baselines into our simulator for a~common ground of comparison. We additionally implemented the packet lookup function for Hicuts, Hypercuts, and Efficuts\footnote{The original source code of Hicuts, Hypercuts, and Efficuts does not implement a~packet lookup function, and instead estimates by the worst case: the maximum depth of the tree.}. In the implementation of MRF-Fast we store the dependencies as a~linked-list at each node. We did not apply the dependency structure minimization (thus, the actual memory consumption by MRF-Fast may be lower). We ran our simulations on a~server with 40 CPU cores (Intel(R) Xeon(R) Gold 6209U CPU @ 2.10GHz) and 192GB RAM. Specifically, each pointer consumes $4$ bytes of memory.

\paragraph*{Metrics:}
We report two main metrics of interest \first average classification time measured as the average number of traversed nodes and \second memory consumed measured in KiloBytes (KB).

In measuring the nodes traversed, we count the number of nodes accessed during lookup and the number of swapped nodes during rearrangement.
For all our baselines, we only count the number of nodes accessed during lookup since they do not perform any rearrangements.
For MRF, we simply add the nodes traversed during lookup and rearrangement: the cost of a~packet match is comparable to the cost of operations that accompany the swap (checking a~rule overlap).

\begin{figure}[!ht]
\centering
\includegraphics[width=0.4\linewidth]{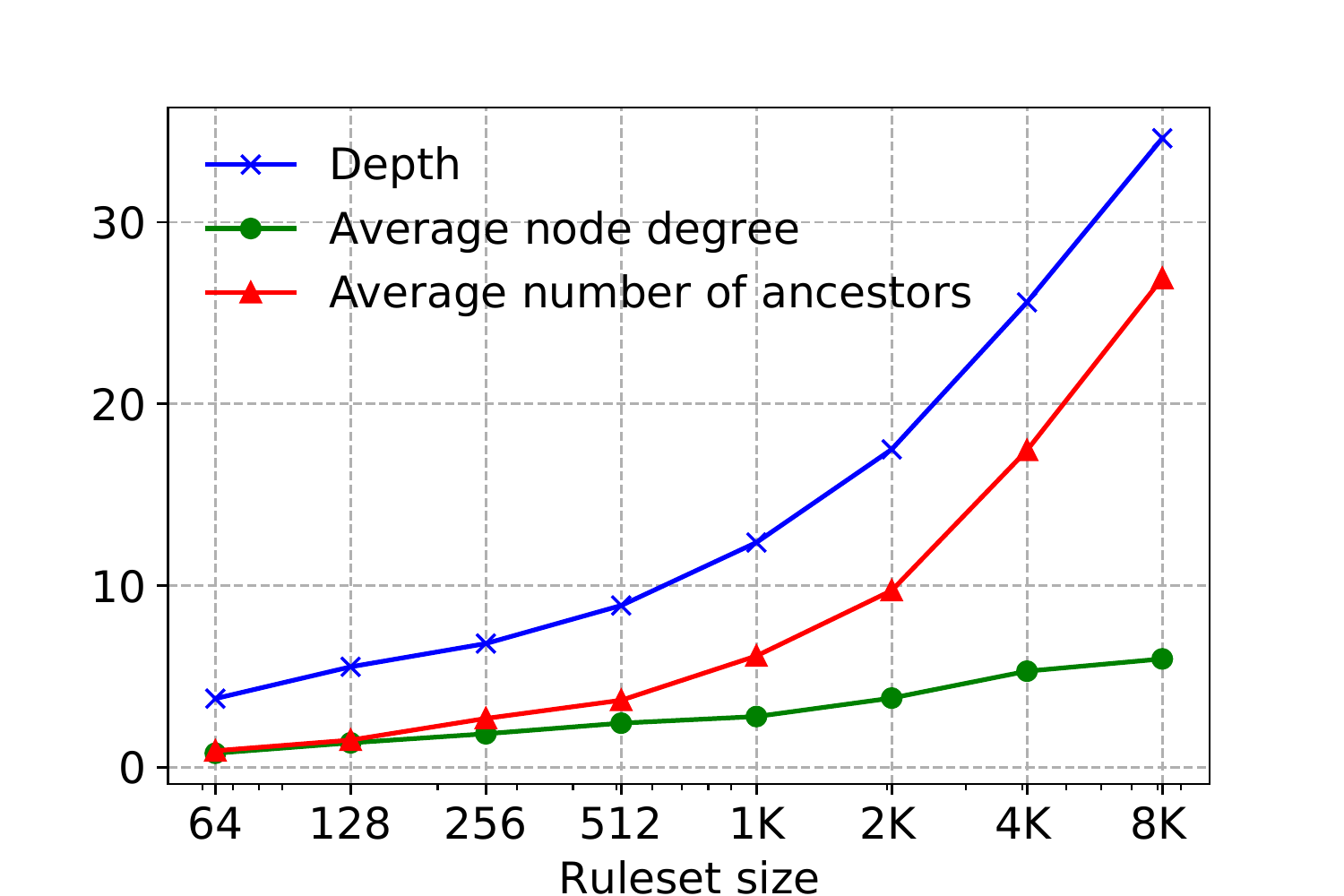}
\caption{Increase of maximum node depth, average node degree, and the average number of ancestors in the dependency graph increase with the ruleset size.}
\label{fig:dag-stats}
\end{figure}

For MRF-Fast, we divide the rearrangement cost by a~factor $\alpha=5$ (see also Section~\ref{sec:non-uniform}): a~packet match requires one comparison with each of the $5$-tuple header fields, whereas the operations that accompany the swap requires only one comparison.
All the values reported in our evaluation are averaged over $32$ runs.

\subsection{Results}
Before presenting our results, we analyze the characteristics of the rulesets used in our evaluations.
Specifically, we are interested in the diversity of the dependency graph's structure across ruleset sizes.
In Figure~\ref{fig:dag-stats}, we fix acl1\_seed provided by Classbench and generate rulesets of sizes in the range $64$ to $8192$.
We observe that some parameters increase with the ruleset size: maximum depth of nodes, average node degree, and the average number of ancestors.
All three metrics of the dependency graph structure influence the classification time.
The larger the average number of ancestors, the larger the classification time of Move-Recursively-Forward, at least for the memoryless implementation: the \emph{direct dependency} operation traverses the entire list of ancestors of a~moving node.

\begin{figure*}[h!]
  \center
  \begin{subfigure}{0.9\linewidth}
  \centering
  \includegraphics[width=1\textwidth]{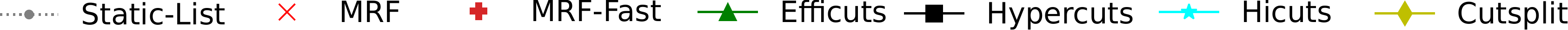}
  \end{subfigure}
  \begin{subfigure}{1\linewidth}
  \centering
  \includegraphics[width=0.47\textwidth]{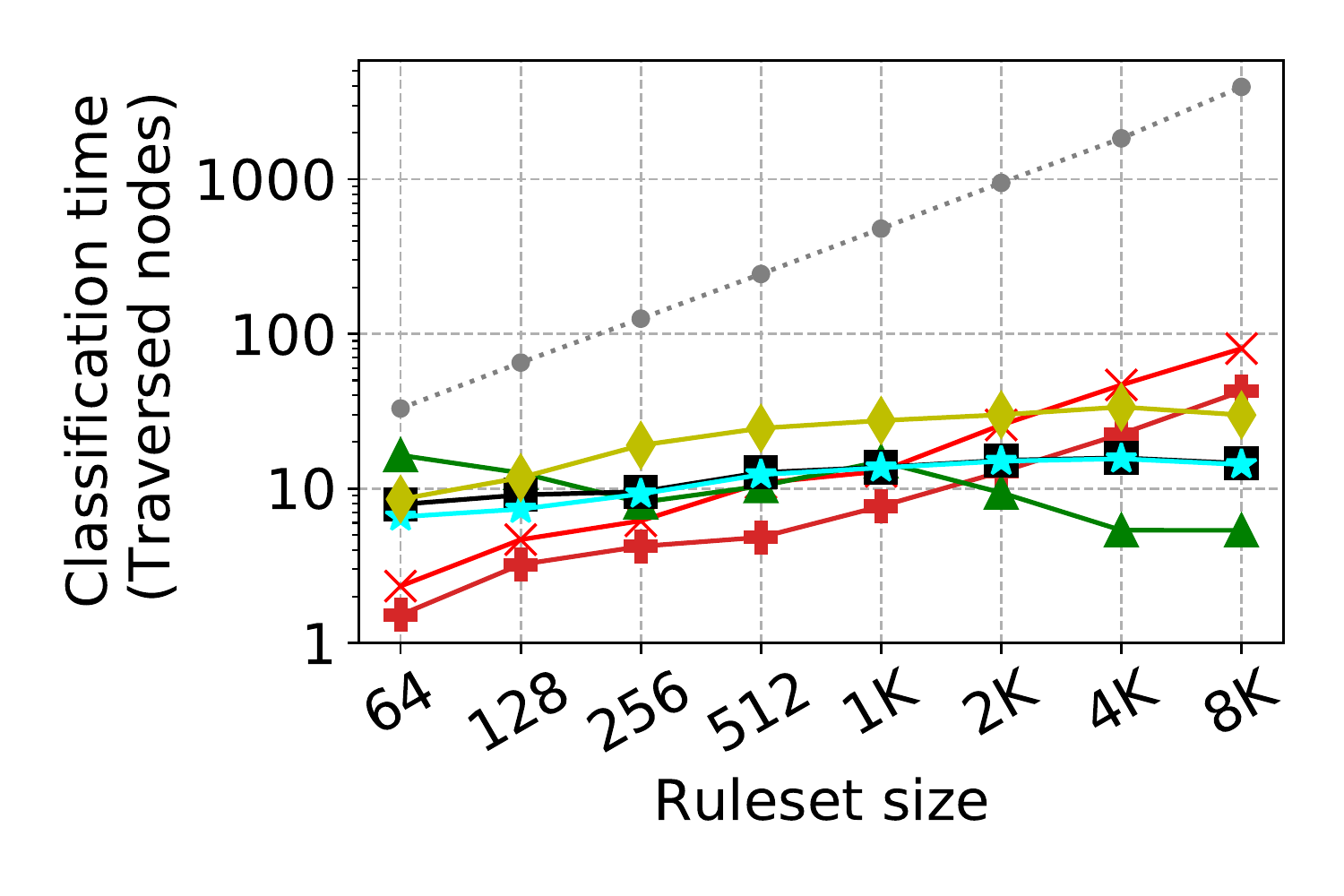}
  \caption{Classification time under high-locality in traffic}
  \label{fig:size-high-locality}
  \end{subfigure}\hfill
  \begin{subfigure}{0.48\linewidth}
  \centering
  \includegraphics[width=0.9\textwidth]{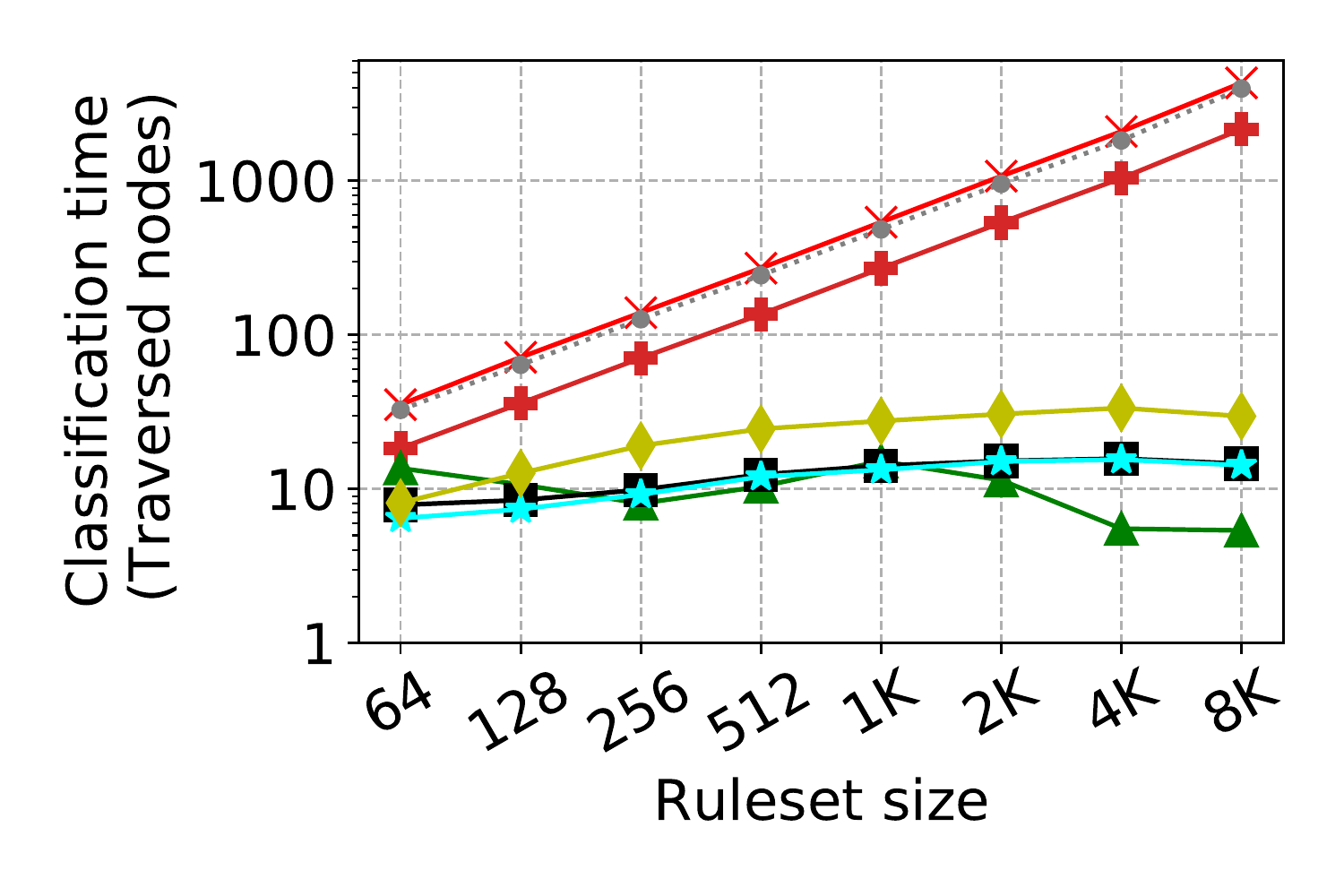}
  \caption{Classification time under low-locality in traffic}
  \label{fig:size-low-locality}
  \end{subfigure}\hfill
  \begin{subfigure}{0.48\linewidth}
  \centering
  \includegraphics[width=0.9\textwidth]{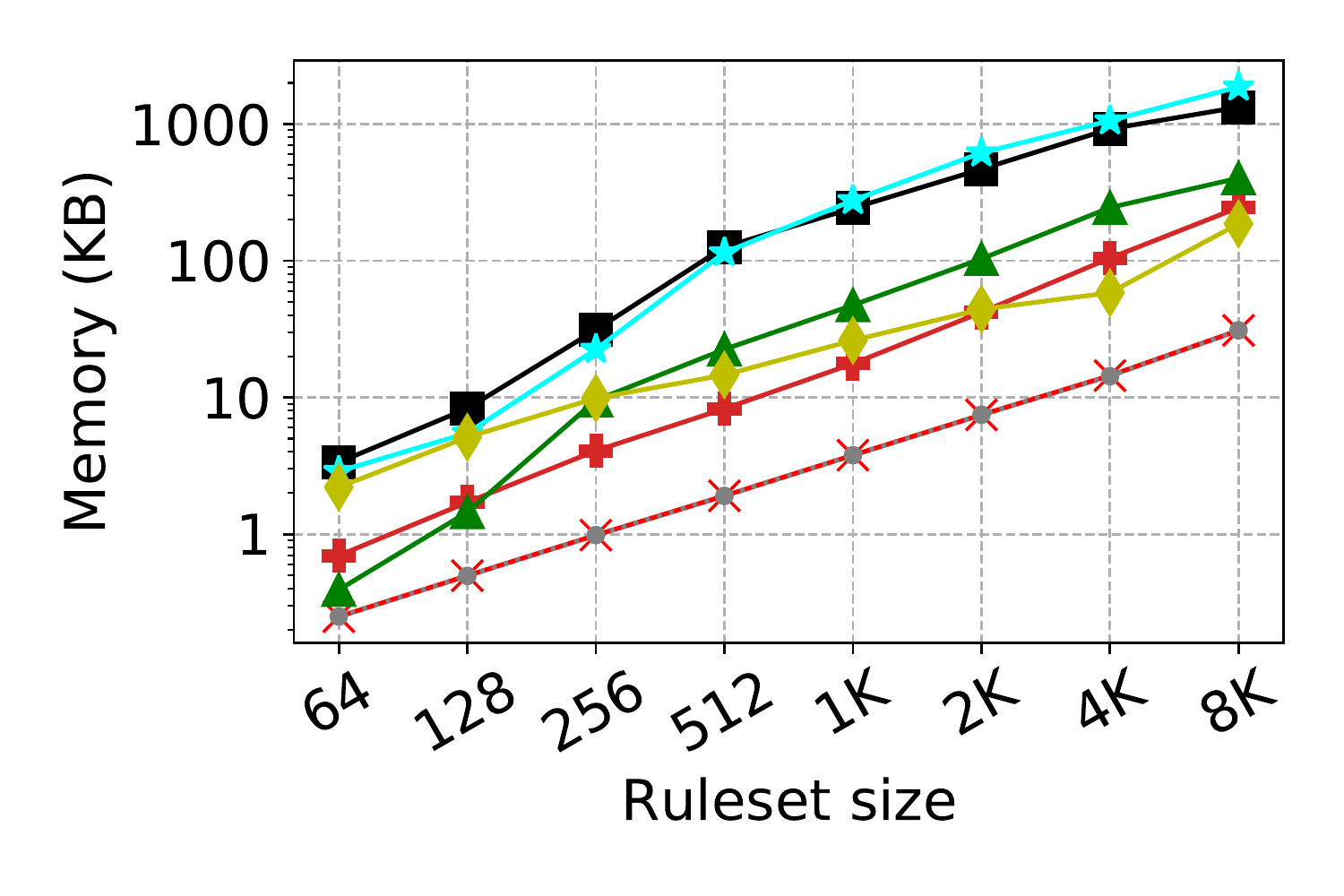}
  \caption{Memory usage: $10$x lower memory consumption by MRF}
  \label{fig:size-memory}
  \end{subfigure}
  \caption{The self-adjusting list-based packet classifier \MRF outperforms decision tree-based algorithms under high locality in traffic for ruleset sizes up to $1$K and significantly improves the memory requirements ($10$x on average). \MRF-Fast trades memory for classification time, showing an improvement over \MRF in classification time with slightly more memory requirements. Note the log scale in the figures.}
\end{figure*}

\paragraph*{MRF under high traffic locality.}
In Figure~\ref{fig:size-high-locality}, we show the average number of nodes traversed with a~high locality in traffic for various ruleset sizes.
For small ruleset sizes in the range $64$ up to $1$K, \MRF outperforms in classification time compared to all our baselines.
For ruleset size of $64$: \MRF performs $7.01$x better than Efficuts, $3.64$x better than CutSplit, and $14.06$x better than a~static-list.
For ruleset size of $1$K: MRF performs $1.15$x better than Efficuts, $2.12$x better than CutSplit, and $37.04$x better than a~static-list.
For larger ruleset sizes ($>1K$): \MRF performs $15$x worse compared to Efficuts and $2.7$x worse compared to CutSplit, yet $49$x better that a~static-list.

For a~small ruleset size (e.g., $64$), we can see from Figure~\ref{fig:dag-stats} that the average number of ancestors is much lower.
This allows \MRF to move the frequently matched rules closer to the head of the list, which significantly improves classification time under high locality.
For large ruleset sizes ($>1$K) the average node ancestors grow up to $30$, which does not allow moving the frequently matched rules closer to the head (a rule cannot be moved ahead of dependencies).

\paragraph*{MRF under low traffic locality.}
We evaluate the performance of \MRF under low locality in traffic even though \MRF's design is not targeted for this case.
A trace with low traffic locality in our evaluations consists of unique packets, i.e., a~packet arrives only once in a~trace, and each packet matches a~rule in the ruleset uniformly at random.
As a~result, the number of nodes that \MRF traverses on average is nearly half the ruleset size.
 In Figure~\ref{fig:size-low-locality}, we observe that the classification time of \MRF is comparable to a~static list for all ruleset sizes. On the other hand, \MRF-Fast slightly improves the classification time compared to MRF but still performs worse compared to decision tree-based classifiers.

\paragraph*{MRF memory consumption.} \MRF significantly reduces memory requirements for packet classification. The memory requirement of \MRF reaches the lower bound for uncompressed data structures (see Section~\ref{sec:efficient}). In Figure~\ref{fig:size-memory} we show the memory requirements of different packet classifiers across various ruleset sizes generated using acl1\_seed.
We observe that \MRF requires on average across various ruleset sizes $10.24$x lower memory compared to Efficuts, $43.95$x lower memory compared to Hypercuts, and $7.58$x lower memory compared to CutSplit.
In Figure~\ref{fig:memory-256} we further evaluate the memory requirements across all the ruleset seeds available in Classbench. We observe significant improvement in memory consumption while the classification time of \MRF is on-par compared to Efficuts and CutSplit as shown in Figure~\ref{fig:ct-256}.

\begin{figure*}
  \center
  \begin{subfigure}{1\linewidth}
  \centering
  \includegraphics[width=0.4\textwidth]{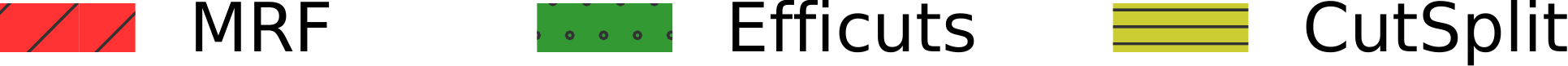}
  \end{subfigure}
  \begin{subfigure}{0.98\linewidth}
  \centering
  \includegraphics[width=1\textwidth]{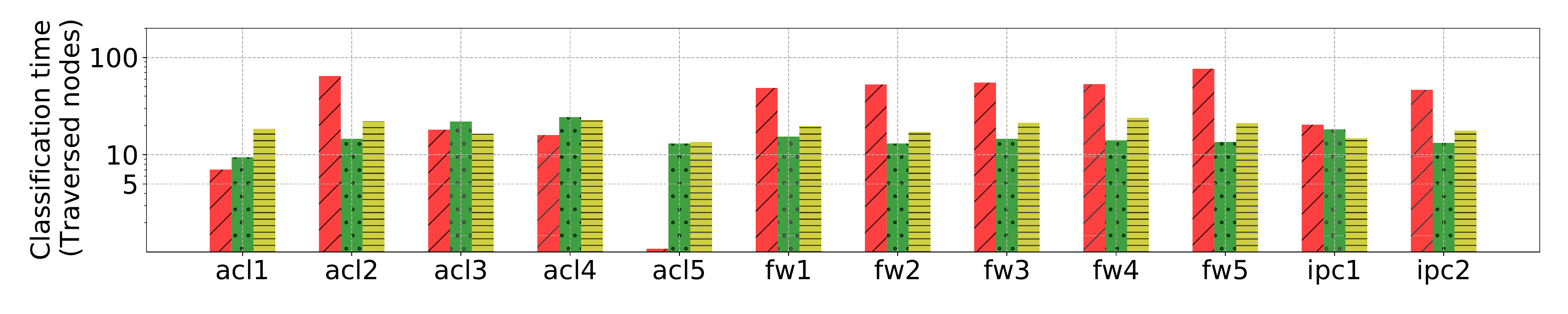}
  \caption{Classification time under high-locality in traffic and ruleset size $256$.}
  \label{fig:ct-256}
  \end{subfigure}
  \begin{subfigure}{0.98\linewidth}
  \centering
  \includegraphics[width=1\textwidth]{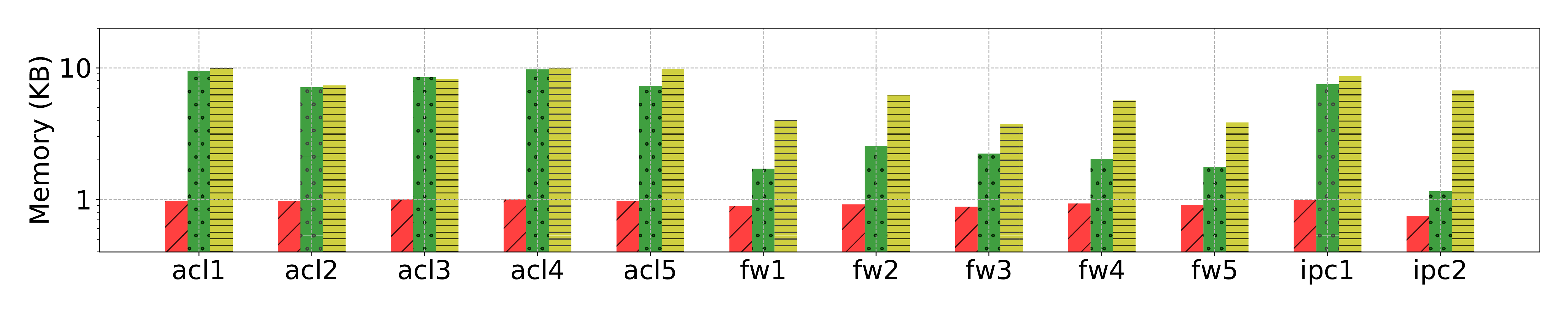}
  \caption{Memory usage for ruleset size $256$.}
  \label{fig:memory-256}
  \end{subfigure}
  \caption{For small rulesets, MRF's classification time is on-par compared to Efficuts but with a~significantly lower memory usage. Note the log scale in the figures.}
\end{figure*}

\section{Conclusions and Future Directions}
\label{sec:conclusions}

We proposed a~packet classifier design that adjusts to traffic and performs provably well. The viability of the solution improves as the locality of traffic increases or the ruleset size decreases, as also demonstrated empirically.

This paper leaves open several research avenues. In particular, it would be interesting to evaluate our drop-in replacement for static lists of rules in existing systems (e.g., 5G packet detection rules), and in particular in existing packet classifiers (e.g., in leaves of decision trees). It would also be interesting to refine our implementation, especially around the \emph{direct dependency} operation, and study hybrid system designs, i.e., self-adjusting lists which may offload existing packet classifiers upon detecting that traffic exhibits high locality.
Design opportunities may also arise for self-adjusting \emph{trees} in packet classification.

To strengthen the theoretical analysis, we point beyond the pure competitive analysis of Move-Recursively-Forward.
In the classic list access setting, the algorithms perform provably better under locality of input (measured in \emph{runs}, a~metric specific to list access)~\cite{Albers16}, and empirical evidence collected in our work similarly suggests that our algorithm may improve too.
In the \emph{average-case analysis} for the stochastic setting, the competitive ratio of Move-To-Front improves from $4$ to $\pi$~\cite{GonnetMS79} (in the paid exchange model), a~behavior that the Move-Recursively-Forward algorithm may share.

\bibliographystyle{plain}
\bibliography{references}

\balance

\end{document}

%% file: macros.tex
\newcommand{\DAG}{\textsc{G}\xspace}
\newcommand{\OPT}{\textsc{OPT}\xspace}
\newcommand{\ALG}{\textsc{ALG}\xspace}

\newcommand{\pos}{\textsf{pos}}
\newcommand{\dep}{\ensuremath{\textbf{d}}}
\newcommand{\depx}{\ensuremath{\delta}}

\newcommand{\Id}{\ensuremath{I^{-}}}
\newcommand{\Ic}{\ensuremath{I^{+}}}

\newcommand{\MRFwd}{\textsc{Move-Recursively-Forward}\xspace}
\newcommand{\MRF}{\textsc{MRF}\xspace}
\newcommand{\MRFx}[1]{\textsc{MRF}(#1)\xspace}

\newcommand{\first}{\emph{(i)}\xspace}
\newcommand{\second}{\emph{(ii)}\xspace}

\newtheoremstyle{sltheorem}
{}                
{}                
{\slshape}        
{}                
{\bfseries}       
{.}               
{ }               
{}                
\theoremstyle{sltheorem}
\newtheorem{theorem}{Theorem}
\newtheorem{lemma}[theorem]{Lemma}